\documentclass[10pt]{llncs}

\usepackage[]{geometry}
\newgeometry
{
    left=2.0cm,
    right=2.0cm,
    top=3.2cm,
    bottom=3.2cm
}

\usepackage{amsmath}
\usepackage{amsfonts}
\usepackage{amssymb}

\spnewtheorem{observation}{Observation}{\bfseries}{\itshape}
\usepackage[ruled, vlined, linesnumbered]{algorithm2e}

\usepackage[utf8]{inputenc}
\usepackage{microtype}
\input{glyphtounicode}
\usepackage{graphicx}

\usepackage{enumerate}

\newcommand{\calB}{\mathcal{B}}
\newcommand{\calC}{\mathcal{C}}
\newcommand{\calD}{\mathcal{D}}

\newcommand{\calO}{\mathcal{O}}
\newcommand{\calP}{\mathcal{P}}

\newcommand{\calT}{\mathcal{T}}

\title
{%
    Minimum Eccentricity Shortest Paths in some Structured Graph Classes%
    \thanks
    {%
        Results of this paper were partially presented at WG~2015\,\cite{DraganLeiter2015}.
    }
}

\author
{%
    Feodor F. Dragan 
    \and 
    Arne Leitert
}

\institute{
    Department of Computer Science, \\
    Kent State University, Kent, Ohio, USA  \\
    \email{dragan@cs.kent.edu}, 
    \email{aleitert@cs.kent.edu} 
}

\makeatletter
\newcommand{\ie}{i.\,e.\@ifnextchar{,}{}{~}}
\newcommand{\eg}{e.\,g.\@ifnextchar{,}{}{~}}
\makeatother

\DeclareMathOperator{\pg}{pg}

\DeclareMathOperator{\ld}{ld}
\DeclareMathOperator{\ecc}{ecc}
\DeclareMathOperator{\diam}{diam}

\DeclareRobustCommand{\qedClaim}
{%
  \ifmmode \lozenge%
  \else%
    \leavevmode\unskip\penalty9999 \hbox{}\nobreak\hfill%
    \quad\hbox{$\lozenge$}%
  \fi%
}%

\SetKw{KwDownTo}{downto}

\begin{document}
\pagestyle{plain}
\maketitle

\begin{abstract}
We investigate the \emph{Minimum Eccentricity Shortest Path} problem in some structured graph classes. 
It asks for a given graph to find a shortest path with minimum eccentricity. 
Although it is NP-hard in general graphs, we demonstrate that a minimum eccentricity shortest path can be found in linear time for distance-hereditary graphs (generalizing the previous result for trees) and give a generalised approach which allows to solve the problem in polynomial time for other graph classes.
This includes chordal graphs, dually chordal graphs, graphs with bounded tree-length, and graphs with bounded hyperbolicity.
Additionally, we give a simple algorithm to compute an additive approximation for graphs with bounded tree-length and graphs with bounded hyperbolicity.
\end{abstract}

\section{Introduction}

The \emph{Minimum Eccentricity Shortest Path} problem asks for a given graph $G=(V,E)$ to find a shortest path~$P$ such that for each other shortest path~$Q$, $\ecc(P) \leq \ecc(Q)$ holds.
Here, the eccentricity of a set $S\subseteq V$ in $G$ is $\ecc(S)=\max_{u \in V} d_G(u, S)$.
This problem was introduced in~\cite{DrLei2015}.
It may arise in determining a ``most accessible'' speedy linear route in a network  and can find applications in communication networks, transportation planning, water resource management and fluid transportation.
It was also shown in~\cite{DKL2014,DrLei2015} that a minimum eccentricity shortest path plays a crucial role in obtaining the best to date approximation algorithm for a minimum distortion embedding of a graph into the line.
Specifically, every graph~$G$ with a shortest path of eccentricity~$r$ admits an embedding~$f$ of $G$ into the line with distortion at most $(8r +  2) \ld(G)$, where $\ld(G)$ is the minimum line-distortion of $G$ (see~\cite{DrLei2015} for details).
Furthermore, if a shortest path of $G$ of eccentricity~$r$ is given in advance, then such an embedding~$f$ can be found in linear time.
Note also that every graph has a shortest path of eccentricity at most $\lfloor \ld(G) / 2 \rfloor$.

Those applications motivate investigation of the Minimum Eccentricity Shortest Path problem  in general graphs and in particular graph classes.
Fast algorithms for it will imply fast approximation algorithms for the minimum line distortion problem.
Existence of low eccentricity shortest paths in structured graph classes will imply low approximation bounds for those classes.
For example, all AT-free graphs (hence, all interval, permutation, cocomparability graphs) enjoy a shortest path of eccentricity at most 1~\cite{COS-SICOMP}, all convex bipartite graphs enjoy a shortest path of eccentricity at most 2~\cite{DKL2014}.

In \cite{DrLei2015}, the Minimum Eccentricity Shortest Path problem was investigated in general graphs.
It was shown that its decision version is NP-complete (even for graphs with vertex degree at most 3).
However, there are efficient  approximation algorithms: a 2-approximation, a 3-approximation, and an 8-approximation for the problem can be computed in $\calO(n^3)$ time, in $\calO(nm)$ time, and in linear time, respectively.
Furthermore, a shortest path of minimum eccentricity~$r$ in general graphs can be computed in $\calO(n^{2r+2}m)$ time. 
Paper~\cite{DrLei2015} initiated also the study of the Minimum Eccentricity Shortest Path problem in special graph classes by showing that a minimum eccentricity shortest path in trees can be found in linear time.
In fact, every diametral path of a tree is a minimum eccentricity shortest path. 

In this paper, we design efficient algorithms for the Minimum Eccentricity Shortest Path problem in distance-hereditary graphs, in chordal graphs, in dually chordal graphs, and in more general graphs with bounded tree-length or with bounded hyperbolicity.
Additionally, we give a simple algorithm to compute an additive approximation for graphs with bounded tree-length and graphs with bounded hyperbolicity.

Note that our Minimum Eccentricity Shortest Path problem is close but different from the \emph{Central Path} problem in graphs introduced in~\cite{Slater}. 
It asks for a given graph~$G$ to find a path~$P$ (not necessarily shortest) such that any other path of $G$ has eccentricity at least $\ecc(P)$.
The Central Path problem generalizes the Hamiltonian Path problem and therefore is NP-hard even for chordal graphs \cite{haiko}.
Our problem is polynomial time solvable for chordal graphs. 

\section{Notions and Notations}
All graphs occurring in this paper are connected, finite, unweighted, undirected, loopless and without multiple edges. For a graph~$G = (V, E)$, we use $n = |V|$ and $m = |E|$ to denote the cardinality of the vertex set and the edge set of~$G$.
$G[S]$ denotes the \emph{induced subgraph} of $G$ with the vertex set~$S$. 

The \emph{length} of a path from a vertex~$v$ to a vertex~$u$ is the number of edges in the path. 
The \emph{distance}~$d_G(u,v)$ of two vertices $u$ and~$v$ is the length of a shortest path connecting $u$ and~$v$. 
The distance between a vertex~$v$ and a set~$S \subseteq V$ is defined as $d_G(v, S) = \min_{u \in S}  d_G(u, v)$. 
The \emph{eccentricity}~$\ecc(v)$ of a vertex~$v$ is $\max_{u \in V} d_G(u, v)$.
For a set~$S \subseteq V$, its eccentricity is $\ecc(S) = \max_{u \in V} d_G(u, S)$.
For a vertex pair $s,t$, a shortest $(s,t)$-path~$P$ has \emph{minimal eccentricity}, if there is no shortest $(s,t)$-path~$Q$ with $\ecc(Q) < \ecc(P)$.
Two vertices $x$ and $y$ are called \emph{mutually furthest} if $d_G(x,y) = \ecc(x) = \ecc(y)$. 
A vertex~$u$ is \emph{$k$-dominated} by a vertex~$v$ (by a set~$S \subset V$), if $d_G(u,v) \leq k$ ($d_G(u, S) \leq k$, respectively).

The \emph{diameter} of a graph~$G$ is $\diam(G) = \max_{u,v \in V} d_G(u, v)$. 
The diameter~$\diam_G(S)$ of a set~$S \subseteq V$ is defined as $\max_{u,v \in S} d_G(u, v)$. 
A pair of vertices $x,y$ of $G$ is called a \emph{diametral  pair} if $d_G(x, y) = \diam(G)$.
In this case, every shortest path connecting $x$ and~$y$ is called a \emph{diametral path}. 

For a vertex~$v \in V$, $N(v) = \{ u \in V \mid uv \in E \}$ is called the \emph{open neighborhood}, and $N[v] = N(v) \cup \{ v \}$ the \emph{closed neighborhood} of $v$. 
$N^r[v] = \{ u \in V \mid d_G(u,v) \leq r \}$ denotes the \emph{disk} of radius~$r$ around vertex~$v$. 
Additionally, $L_r^{(v)} = \{ u \in V \mid d_G(u,v) = i \}$ denotes the vertices with distance~$r$ from~$v$. 
For two vertices $u$ and~$v$, $I(u,v) = \{ w \mid d_G(u,v) = d_G(u,w) + d_G(w,v) \}$ is the \emph{interval} between $u$ and~$v$. 
The set~$S_i(s,t) = L_i^{(s)} \cap I(u,v)$ is called a \emph{slice} of the interval from $u$ to~$v$. 
For any set~$S \subseteq V$ and a vertex~$v$, $\Pr(v, S) = \{ u \in S \mid d_G(u, v) = d_G(v, S) \}$ denotes the \emph{projection} of $v$ on~$S$. 

A \emph{chord} in a path is an edge connecting two non-consecutive vertices of the path. 
A set of vertices~$S$ is a \emph{clique} if all vertices in $S$ are pairwise adjacent. 
A graph is \emph{chordal} if every cycle with at least four vertices has a chord. 
A graph is \emph{distance-hereditary} if the distances in any connected induced subgraph are the same as they are in the original graph. 
A graph is \emph{dually chordal} if it is the intersection graph of maximal cliques of a chordal graph.
For more definitions of these classes and relations between them see~\cite{BrLeSpinGraphClasses}.

\section{A Linear-Time Algorithm for Distance-Hereditary Graphs}
    \label{sec:DistHered}

\emph{Distance-hereditary graphs} can be defined as graphs where each chordless path is a shortest path~\cite{howorka}. 
Several interesting characterizations of distance-hereditary graphs in terms of metric and neighborhood properties, and forbidden configurations were provided by \textsc{Bandelt} and \textsc{Mulder}~\cite{BM-dhg}, and by \textsc{D'Atri} and \textsc{Moscarini}~\cite{D-AM-dhg}. 
The following proposition  lists  the basic information on distance-hereditary graphs that is needed in what follows. 

\begin{proposition}
    [\cite{BM-dhg,D-AM-dhg}] 
    \label{prop:dhg} 
For a graph~$G$ the following conditions are equivalent:
\begin{enumerate}[(1)]
    \item
        $G$ is distance-hereditary;
    \item
        The house, domino, gem (see Fig.~\ref{fig:dhg}) and the cycles~$C_k$ of length~$k \geq 5$ are not induced subgraphs of $G$;
    \item
        For an arbitrary vertex~$x$ of $G$ and every pair of vertices $u, v \in L_k^{(x)}$, that are in the same connected component of the graph $G[V \setminus L_{k-1}^{(x)}]$, we have $N(v) \cap L_{k-1}^{(x)} = N(u) \cap L_{k-1}^{(x)}$.
    \item
        (4-point condition)
        For any four vertices $u, v, w, x$ of $G$ at least two of the following distance sums are equal:
        $d_G(u, v) + d_G(w, x)$; $d_G(u, w) + d_G(v, x)$; $d_G(u, x) + d_G(v, w)$.
        If the smaller sums are equal, then the largest one exceeds the smaller ones at most by 2.
\end{enumerate}
\end{proposition}

\begin{figure}
    [htb]
    \centering
    \includegraphics[]{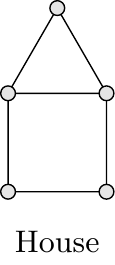}%
    \hspace*{1cm}%
    \includegraphics[]{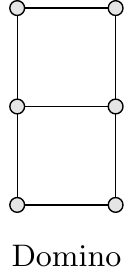}%
    \hspace*{1cm}%
    \includegraphics[]{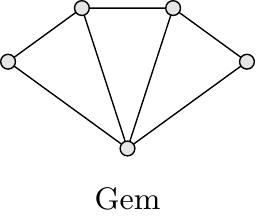}%
    \caption
    {
        Forbidden induced subgraphs in a distance-hereditary graph. 
    }
    \label{fig:dhg} %
\end{figure}

As a consequence of statement (3) of Proposition~\ref{prop:dhg} we get. 

\begin{corollary}
    \label{cor:pr-size} 
Let $P := P(s,t)$ be a shortest path in a distance-hereditary graph~$G$ connecting vertices $s$ and~$t$, and $w$ be an arbitrary vertex of~$G$. 
Let $a$ be a vertex of $\Pr(w, P)$ that is closest to $s$, and let $b$ be a vertex of $\Pr(w,P)$ that is closest to~$t$. 
Then $d_G(a,b) \leq 2$ and there must be a vertex $w'$ in $G$ adjacent to both $a$ and~$b$ and at distance $d_G(w,P) - 1$ from~$w$. 
\end{corollary}

As a consequence of statement (4) of Proposition \ref{prop:dhg} we get. 

\begin{corollary}
    \label{cor:two-paths} 
Let $x,y,v,u$ be arbitrary vertices of a distance-hereditary graph~$G$ with $v \in I(x,u)$, $u \in I(y,v)$, and $d_G(u,v) > 1$, then $d_G(x,y) = d_G(x,v) + d_G(v,u) + d_G(u,y)$. 
That is, if two shortest paths share ends of length at least~$2$, then their union is a shortest path. 
\end{corollary} 

\begin{proof}
Consider distance sums $S_1 := d_G(x,v) + d_G(u,y)$, $S_2 := d_G(x,y) + d_G(u,v)$ and $S_3 := d_G(x,u)+d_G(v,y)$. 
Since $d_G(x,u) + d_G(v,y) = d_G(x,v) + d_G(u,y) + 2 \, d_G(u,v)$, we have $S_3 > S_1$. 
Then, either $S_2 = S_3$ or $S_1 = S_2$ and $S_3 - S_1 \leq 2$. 
If the latter is true, then $2 \geq S_3 - S_1 = d_G(x,v) + d_G(u,y) + 2 \, d_G(u,v) - d_G(x,v) - d_G(u,y) = 2 \, d_G(v,u) > 2$ and a contradiction arises.
Thus, $S_2 = S_3$ and we get $d_G(x,y) = d_G(x,v) + d_G(v,u) + d_G(u,y)$.
\qed
\end{proof}

\begin{lemma}
    \label{lm:dhg} 
Let $x,y$ be a diametral pair of vertices of a distance-hereditary graph~$G$, and $k$ be the minimum eccentricity of a shortest path in $G$. 
If for some shortest path $P = P(x,y)$, connecting $x$ and $y$, $\ecc(P) > k$ holds, then $\diam(G) = d_G(x,y) \geq 2k$. 
Furthermore, if $d_G(x,y) = 2k$ then there is a shortest path $P^*$ between $x$ and~$y$ with $\ecc(P^*) = k$. 
\end{lemma} 

\begin{proof} 
Consider a vertex $v$ with $d_G(v,P) > k$. 
Let $x'$ be a vertex of $\Pr(v,P)$ closest to $x$, and $y'$ be a vertex of $\Pr(v,P)$ closest to $y$. 
By Corollary~\ref{cor:pr-size}, $d_G(x',y') \leq 2$ and there must be a vertex $v'$ in $G$ adjacent to both $x'$ and $y'$ and at distance $d_G(v,P) - 1$ from $v$. 
Let $P(x,x')$ and $P(y',y)$ be subpaths of $P$ connecting vertices $x,x'$ and vertices $y,y'$, respectively. 
Consider also an arbitrary shortest path $Q(v,v')$ connecting $v$ and $v'$ in $G$. 
By choices of $x'$ and $y'$, no chords in $G$ exist in paths $P(x,x') \cup Q(v',v)$ and $P(y,y') \cup Q(v',v)$. 
Hence, those paths are shortest in $G$. 
Since $x,y$ is a diametral pair, we have $d_G(x,x') + d_G(x',y') + d_G(y',y) = d_G(x,y) \geq d_G(x,v) = d_G(x,x') + 1 + d_G(v',v)$. 
That is, $d_G(y',y) \geq d_G(v',v) + 1 - d_G(x',y')$.
Similarly, $d_G(x',x) \geq d_G(v',v)+1-d_G(x',y')$. 
Combining both inequalities and taking into account that $d_G(v,v') \geq k$, we get  $d_G(x,y) = d_G(x,x') + d_G(x',y') + d_G(y',y) \geq 2k + 2 - d_G(x',y') \geq 2k$. 
Furthermore, we have $d_G(x,y) \geq 2k + 1$ if $d_G(x',y') = 1$ and $d_G(x,y) \geq 2k + 2$ if $d_G(x',y') = 0$.
Also, if $d_G(x,y) = 2k$ then $d_G(x',y') = 2$, $d_G(v,v') = k$, $d_G(x,x') = d_G(y,y') = k-1$ and $d_G(v,x) = d_G(v,y) = 2k$.

Now assume that $d_G(x,y) = 2k$. 
Consider sets $S = \{ w \in V \mid d_G(x,w) = d_G(y,w) = k \}$ and $F_{x,y} = \{u \in V \mid d_G(u,x) = d_G(u,y) = 2k \}$. 
Let $c \in S$ be a vertex of $S$ that $k$-dominates the maximum number of vertices in $F_{x,y}$. 
Consider a shortest path $P^*$ connecting vertices $x$ and $y$ and passing through vertex $c$. 
We will show that $\ecc(P^*) = k$. 
Let $x'$ ($y'$) be the neighbor of $c$ in subpath of $P^*$ connecting $c$ with $x$ (with $y$, respectively).

Assume there is a vertex $v$ in $G$ such that $d_G(v,P^*) > k$. 
As in the first part of the proof, one can show that $d_G(v,x') = d_G(v,y') = k + 1$, \ie, $x',y' \in \Pr(v,P^*)$ and $d_G(v,P^*) = k + 1$. 
Furthermore, $d_G(v,x) = d_G(v,y) = 2k$, \ie, $v \in F_{x,y}$. 
Also, vertex $v'$, that is adjacent to $x'$, $y'$ and at distance $k$ from $v$, must belong to $S$. 
Since $d_G(v,c) > k$ but $d_G(v,v') = k$, by choice of $c$, there must exist a vertex $u \in F_{x,y}$ such that $d_G(u,c) \leq k$ and $d_G(u,v') > k$. 
Since $d_G(u,y) = d_G(u,x) = 2k$, $d_G(u,c)$ must equal $k$ and both $d_G(u,x')$ and $d_G(u,y')$ must equal $k + 1$.  

Since $d_G(v,u) \leq \diam(G) = 2k$ and $d_G(v,y') = d_G(v,x') = k + 1 = d_G(u,x') = d_G(u,y')$, we must have a chord between vertices of a shortest path $P(v,v')$ connecting $v$ with $v'$ and vertices of a shortest path $P(u,c)$ connecting $u$ with $c$. 
If no chords exist or only chord $cv'$ is present, then $d_G(v,u)\geq 2k+1$, contradicting with $\diam(G) = 2k$. 
So, consider a chord $ab$ with $a \in P(v,v')$, $b \in P(u,c)$, $ab \neq cv'$, and $d_G(a,v') + d_G(b,c)$ is minimum. 
We know that $d_G(a,v') = d_G(b,c)$ must hold since $d_G(u,v') > k = d_G(u,c)$ and $d_G(v,c) > k = d_G(v,v')$.
To avoid induced cycles of length $k \geq 5$, $d_G(a,v') = d_G(b,c) = 1$ must hold.
But then, vertices $a,b,c,x',v'$ form either an induced cycle~$C_5$, when $c$ and $v'$ are not adjacent, or a house, otherwise.
Note that, by distance requirements, edges $bv'$, $ca$, $bx'$, and $ax'$ are not possible. 

Contradictions obtained show that such a vertex $v$ with $d_G(v,P^*) > k$ is not possible, \ie, $\ecc(P^*) = k$. 
\qed
\end{proof}

\begin{lemma}
    \label{lm:mf} 
In every distance-hereditary graph there is a minimum eccentricity shortest path $P(s,t)$ where $s$ and~$t$ are two mutually furthest vertices. 
\end{lemma}

\begin{proof} 
Let $k$ be the minimum eccentricity of a shortest path in $G$. 
Let $Q := Q(s,t) = (s = v_0, v_1, \ldots, v_i, \ldots, v_q = t)$ be a shortest path of $G$ of eccentricity~$k$ with maximum~$q$, that is, among all shortest paths with eccentricity~$k$, $Q$ is a longest one. 
Assume, without loss of generality, that $t$ is not a vertex most distant from $s$. 
Let $i \leq q$ be the smallest index such that subpath $Q(s,v_i) = (v_0, v_1, \ldots, v_i)$ of $Q$ has also the eccentricity~$k$. 
By choice of $i$, there must exist a vertex $v$ in $G$ which is $k$-dominated only by vertex $v_i$ of $Q(s,v_i)$, \ie, $\Pr(v,Q(s,v_i)) = \{ v_i \}$ and $d_G(v,Q(s,v_i)) = k$. 
Let $P(v,v_i)$ be an arbitrary shortest path of $G$ connecting $v$ with $v_i$.
By choice of $i$, no vertex of $P(v,v_i) \setminus \{ v_i \}$ is adjacent to a vertex of $Q(s,v_i) \setminus \{ v_i \}$. 
Hence, path obtained by concatenating $Q(s,v_i)$ with $P(v_i,v)$ is chordless and, therefore, shortest in $G$, and has eccentricity~$k$, too. 
Note that $v$ is now a most distant vertex from $s$, \ie, $d_G(s,v) = \ecc(s)$. 
Since $d_G(s,v) > d_G(s,t)$, a contradiction with maximality of $q$ arises. 
\qed
\end{proof}

The main result of this section is the following.

\begin{theorem}
    \label{th:dhg} 
Let $x,y$ be a diametral pair of vertices of a distance-hereditary graph~$G$, and $k$ be the minimum eccentricity of a shortest path in $G$. 
Then, there is a shortest path~$P$ between $x$ and~$y$ with $\ecc(P) = k$.  
\end{theorem}

\begin{proof} 
We may assume that for some shortest path $P'$ connecting $x$ and $y$, $\ecc(P') > k$ holds (otherwise, there is nothing to prove).
Then, by Lemma~\ref{lm:dhg}, we have $d_G(x,y) \geq 2k$.

Let $Q := Q(s,t) = (s = v_0, v_1, \ldots, v_i, \ldots, v_q = t)$ be a shortest path of $G$ of eccentricity $k$ such that $s$ and $t$ are two mutually furthest vertices (see Lemma~\ref{lm:mf}). 
Consider projections of $x$ and $y$ to $Q$. 
We distinguish between three cases: 
$\Pr(x,Q)$ is completely on the left of $\Pr(y,Q)$ in $Q$; $\Pr(x,Q)$ and $\Pr(y,Q)$ have a common vertex $w$; and the remaining case (see Corollary~\ref{cor:pr-size}) when $\Pr(x,Q) = \{ v_{i-1}, v_{i+1} \}$ and $\Pr(y,Q) = \{ v_i \}$ for some index $i$.

\medskip

\noindent
\emph{Case 1: $\Pr(x,Q)$ is completely on the left of $\Pr(y,Q)$ in $Q$.}    

\medskip
\noindent
Let $x'$ be a vertex of $\Pr(x,Q)$ closest to $t$ and $y'$ a vertex of $\Pr(y,Q)$ closest to $s$. 
Consider an arbitrary shortest path $P(x,x')$ of $G$ connecting vertices $x$ and $x'$, an arbitrary shortest path $P(y',y)$ of $G$ connecting vertices $y'$ and $y$, and a subpath $Q(x',y')$ of $Q(s,t)$ between vertices $x'$ and $y'$. 
We claim that the path $P$ of $G$ obtained by concatenating $P(x,x')$ with $Q(x',y')$ and then with $P(y',y)$ is a shortest path of eccentricity~$k$. 

Indeed, by choice of $x'$, no edge connecting a vertex in $P(x,x') \setminus \{ x' \}$ with a vertex in $Q(x',y') \setminus \{ x' \}$ can exist in $G$. 
Similarly, no edge connecting a vertex in $P(y',y) \setminus \{ y' \}$ with a vertex in $Q(x',y') \setminus \{ y' \}$ can exist in $G$. 
Since we also have $d_G(x,y) \geq 2k$, $d_G(x,Q) \leq k$ and $d_G(y,Q) \leq k$, no edge connecting a vertex in  $P(y',y) \setminus \{ y' \}$ with a vertex in $P(x,x') \setminus \{ x' \}$ can exist in $G$. 
Hence, chordless path $P = P(x,x') \cup Q(x',y') \cup P(y',y)$ is a shortest path of $G$.

Consider now an arbitrary vertex $v$ of $G$. 
We want to show that $d_G(v,P) \leq k$. 
Since $\ecc(Q) = k$, $d_G(v,Q) \leq k$. 
Consider the projection of $v$ to $Q$. 
We may assume that $\Pr(v,Q) \cap Q(x',y') = \emptyset$ and, without loss of generality, that vertices of $\Pr(v,Q)$ are closer to $s$ than vertex $x'$. 
Let $v'$ be a vertex of $\Pr(v,Q)$ closest to $x'$.
As before, by choices of $v'$ and $y'$, paths $P(y,y') \cup Q(y',v')$ and $P(v,v') \cup Q(y',v')$ are chordless and, therefore, are shortest paths of $G$ (here $P(v,v')$ is an arbitrary shortest path of $G$ connecting $v$ with $v'$). 
Since $d_G(v',y') \geq 2$, by Corollary~\ref{cor:two-paths}, $d_G(v,y) = d_G(v,v') + d_G(v',y') + d_G(y',y)$. 
Hence, from $d_G(x,y) \geq d_G(y,v)$, $d_G(x,y) = d_G(x,x') + d_G(x',y)$ and $d_G(v,y) = d_G(v,x') + d_G(x',y)$, we obtain $d_G(v,x') \leq d_G(x,x') \leq k$.

\medskip
\noindent
\emph{Case 2: $\Pr(x,Q)$ and $\Pr(y,Q)$ have a common vertex $w$.} 

\medskip
\noindent
In this case, we have $d_G(x,y) \leq d_G(x,w) + d_G(y,w) \leq k + k = 2k$. 
Earlier we assumed also that $d_G(x,y) \geq 2k$.
Hence, $\diam(G) = d_G(x,y) = 2k$ and the statement of the theorem follows from Lemma~\ref{lm:dhg}. 

\medskip
\noindent
\emph{Case 3: Remaining case when $\Pr(x,Q) = \{ v_{i-1}, v_{i+1} \}$ and $\Pr(y,Q) = \{ v_i \}$ for some index $i$.}

\medskip
\noindent
In this case, we have $d_G(x,y) \leq d_G(x,v_{i-1}) + 1 + d_G(v_{i},y) \leq 2k + 1$. 
By Lemma~\ref{lm:dhg}, we can assume that $\diam(G) = d_G(x,y) = 2k + 1$, \ie, $d_G(x,v_{i-1}) = d_G(x,v_{i+1}) = d_G(v_{i},y) = k$.

Let $Q(s,v_{i-1})$ and $Q(t,v_{i+1})$ be subpaths of $Q$ connecting vertices $s$ and $v_{i-1}$ and vertices $t$ and $v_{i+1}$, respectively. 
Pick an arbitrary shortest path $P(y,v_{i})$ connecting $y$ with $v_{i}$. 
Since no chords are possible between $Q(s,v_{i}) \setminus \{ v_{i} \}$ and $P(y,v_{i}) \setminus \{ v_{i} \}$ and between $Q(t,v_{i}) \setminus \{ v_{i} \}$ and $P(y,v_{i}) \setminus \{ v_{i} \}$, we have
$d_G(y,t) = d_G(y,v_{i}) + d_G(v_{i},t) = k + d_G(v_{i},t)$ and $d_G(y,s) = d_G(y,v_{i}) + d_G(v_{i},s) = k + d_G(v_{i},s)$.
Inequalities $d_G(x,y) \geq d_G(y,t)$ and $d_G(x,y) \geq d_G(y,s)$ imply $d_G(v_{i+1},t) \leq d_G(v_{i+1},x) = k$ and $d_G(v_{i-1},s) \leq d_G(v_{i-1},x) = k$. 
If both $d_G(v_{i+1},t)$ and $d_G(v_{i-1},s)$ equal $k$, then $d_G(s,t) = 2k + 2$ contradicting with $\diam(G) = 2k+1$. 
Hence, we may assume, without loss of generality, that $d_G(v_{i-1},s) \leq k-1$. 
We will show that shortest path $P := P(x,v_{i+1}) \cup P(v_{i},y)$ has eccentricity $k$ (here, $P(x,v_{i+1})$ is an arbitrary shortest path of $G$ connecting $x$ with $v_{i+1}$). 

Consider a vertex $v$ in $G$ and assume that $\Pr(v,Q)$ is strictly contained in $Q(t,v_{i+1})$. 
Denote by $v'$ the vertex of $\Pr(v,Q)$ that is closest to $s$. 
Let $P(v,v')$ be an arbitrary shortest path connecting $v$ and $v'$. 
As before, $P(v,v') \cup Q(v',s)$ is a chordless path and therefore $d_G(v,s) = d_G(v,v_{i+1}) + d_G(v_{i+1},s)$.
Since $t$ is a most distant vertex from $s$, $d_G(s,v) \leq d_G(s,t)$. 
Hence, $d_G(v,v_{i+1}) + d_G(v_{i+1},s) = d_G(s,v) \leq d_G(s,t) = d_G(s,v_{i+1}) + d_G(v_{i+1},t)$, \ie, $d_G(v,v_{i+1}) \leq d_G(v_{i+1},t) \leq k$. 

Consider a vertex $v$ in $G$ and assume now that $\Pr(v,Q)$ is strictly contained in $Q(s,v_{i-1})$. 
Denote by $v'$ the vertex of $\Pr(v,Q)$ that is closest to $t$. 
Let $P(v,v')$ be an arbitrary shortest path connecting $v$ and $v'$. 
Again, $P(v,v') \cup Q(v',t)$ is a chordless path and therefore $d_G(v,t) = d_G(v,v_{i}) + d_G(v_{i},t)$.
Since $s$ is a most distant vertex from $t$, $d_G(t,v) \leq d_G(s,t)$. 
Hence, $d_G(v,v_{i}) + d_G(v_{i},t) = d_G(t,v) \leq d_G(s,t) = d_G(s,v_{i}) + d_G(v_{i},t)$, \ie, $d_G(v,v_{i}) \leq d_G(v_{i},s) \leq k$. 

Thus, all vertices of $G$ are $k$-dominated by $P(x,v_{i+1}) \cup P(v_{i},y)$. 
\qed
\end{proof}

It is known~\cite{DrNi} that a diametral pair of a distance-hereditary graph can be found in linear time. 
Hence, according to Theorem~\ref{th:dhg}, to find a shortest path of minimum eccentricity in a distance-hereditary graph in linear time, one needs to efficiently extract a best eccentricity shortest path for a given pair of end-vertices. 
In what follows, we demonstrate that, for a distance-hereditary graph, such an extraction can be done in linear time as well. 

We will need few auxiliary lemmas. 

\begin{lemma}
    \label{lem:dhgGateVertex}
In a distance-hereditary graph~$G$, for each pair of vertices $s$ and $t$, if $x$ is on a shortest path from $v$ to $\Pi_v = \Pr(v, I(s,t))$ and $d_G(x,\Pi_v) = 1$, then $\Pi_v \subseteq N(x)$.
\end{lemma}

\begin{proof}
Let $p$ and~$q$ be two vertices in $\Pi_v$ and $d_G(v, \Pi_v) = r$.
By statement~(3) of Proposition~\ref{prop:dhg}, $N(p) \cap L_{r-1}^{(v)} = N(q) \cap L_{r-1}^{(v)}$. 
Thus, each vertex $x$ on a shortest path from $v$ to $\Pi_v$ with $d_G(x,\Pi_v) = 1$ (which is in $N(p) \cap L_{r-1}^{(v)}$ by definition) is adjacent to all vertices in $\Pi_v$, \ie, $\Pi_v \subseteq N(x)$.
\qed
\end{proof}

\begin{lemma}
    \label{lem:dhgSliceJoin}
In a distance-hereditary graph~$G$, let $S_i(s,t)$ and $S_{i+1}(s,t)$ be two consecutive slices of an interval $I(s,t)$. 
Each vertex in $S_i(s,t)$ is adjacent to each vertex in $S_{i+1}(s,t)$.
\end{lemma}

\begin{proof}
Consider statement~(3) of Proposition~\ref{prop:dhg} from perspective of~$t$.
Thus, $S_i(s,t) \subseteq N(v)$ for each vertex~$v \in S_{i+1}(s,t)$. 
Additionally, from perspective of $s$, $S_{i+1}(s,t) \subseteq N(u)$ for each vertex~$u \in S_{i}(s,t)$.
\qed
\end{proof}

\begin{lemma}
    \label{lem:dhgProjSliceInter}
In a distance-hereditary graph~$G$, if a projection $\Pi_v = \Pr(v, I(s,t))$ intersects two slices of an interval $I(s,t)$, each shortest $(s,t)$-path intersects $\Pi_v$.
\end{lemma}

\begin{proof}
Because of Lemma~\ref{lem:dhgGateVertex}, there is a vertex $x$ with $N(x) \supseteq \Pi_v$ and $d_G(v,x) = d_G(v, \Pi_v) - 1$. 
Thus, $\Pi_v$ intersects at most two slices of interval $I(s,t)$ and those slices have to be consecutive, otherwise $x$ would be a part of the interval.
Let $S_i(s,t)$ and $S_{i+1}(s,t)$ be these slices. 
Note that $d_G(s, x) = i + 1$. 
Thus, by statement~(3) of Proposition~\ref{prop:dhg}, $N(x) \cap S_i(s,t) = N(u) \cap S_i(s,t)$ for each $u \in S_{i+1}(s,t)$. 
Therefore, $S_i(s,t) \subseteq \Pi_v$, \ie, each shortest path from $s$ to~$t$ intersects $\Pi_v$.
\qed
\end{proof}

From the lemmas above, we can conclude that, for determining a shortest $(s,t)$-path with minimal eccentricity, a vertex~$v$ is only relevant if $d_G(v, I(s,t)) = \ecc(I(s,t))$ and the projection of $v$ on the interval $I(s,t)$ only intersects one slice.
Algorithm~\ref{algo:MinEccDistHere} uses this.

\begin{algorithm}
    [htb]
    \caption
    {
        \label{algo:MinEccDistHere}
        Computes a shortest $(s,t)$-path~$P$ with minimal eccentricity for a given distance-hereditary graph~$G$ and a vertex pair~$s,t$.
    }

\KwIn
{
    A distance-hereditary graph~$G = (V, E)$ and two distinct vertices $s$ and~$t$.
}

\KwOut
{
    A shortest path~$P$ from $s$ to~$t$ with minimal eccentricity.
}

Compute the sets $V_i = \{ v \mid d_G(v, I(s,t)) = i \}$ for $1 \leq i \leq \ecc(I(s,t))$.

Each vertex $v \notin I(s,t)$ gets a pointer~$g(v)$ initialised with $g(v) := v$ if $v \in V_1$, and $g(v) := \varnothing$ otherwise.

\For
{
    $i := 2$ \KwTo $\ecc(I(s,t))$
    \label{line:dhgFindGateLoop}
}
{
    For each $v \in V_i$, select a vertex~$u \in V_{i-1} \cap N(v)$ and set $g(v) := g(u)$.
    \label{line:dhgFindGateIteration}
}

\ForEach
{
    $v \in V_{\ecc(I(s,t))}$
}
{
    If $N(g(v))$ intersects only one slice of $I(s,t)$, flag $g(v)$ as \emph{relevant}. 
    \label{line:dhgFlagRelevant}
}

Set $P := \{ s, t \}$.

\For
{
    $i := 1$ \KwTo $d_G(s,t) - 1$
    \label{line:dhgSelectPLoop}
}
{
    Find a vertex~$v \in S_i(s,t)$ for which the number of \emph{relevant} vertices in $N(v)$ is maximal.

    Add $v$ to $P$.
    \label{line:dhgSelectPAddV}
}

\end{algorithm}

\begin{lemma}
    \label{lm:path-pair} 
For a distance-hereditary graph~$G$ and an arbitrary vertex pair~$s,t$, Algorithm~\ref{algo:MinEccDistHere} computes a shortest $(s,t)$-path with minimal eccentricity in linear time.
\end{lemma}

\begin{proof}
The loop in line~\ref{line:dhgFindGateLoop} determines for each vertex~$v$ outside of the interval~$I(s,t)$ a \emph{gate vertex}~$g(v)$ such that $N(g(v)) \supseteq \Pr(v, I(s,t))$ and $d_G(v, I(s,t)) = d_G(v, g(v)) + 1$ (see Lemma~\ref{lem:dhgGateVertex}).
From Lemma~\ref{lem:dhgProjSliceInter} and Lemma~\ref{lem:dhgSliceJoin}, it follows that for a vertex~$v$ which is not in $V_{\ecc(I(s,t))}$ or its projection to $I(s,t)$ is intersecting two slices of $I(s,t)$, $d_G(v,P(s,t)) \leq \ecc(I(s,t))$ for every shortest path $P(s,t)$ between $s$ and $t$.
Therefore, line~\ref{line:dhgFlagRelevant} only marks $g(v)$ if $v \in V_{\ecc(I(s,t))}$ and its projection $\Pr(v, I(s,t))$ intersects only one slice. 
Because only one slice is intersected and each vertex in a slice is adjacent to all vertices in the consecutive slice (see Lemma~\ref{lem:dhgSliceJoin}), in each slice the vertex of an optimal (of minimum eccentricity) path $P$ can be selected independently from the preceding vertex.
If a vertex~$x$ of a slice~$S_i(s,t)$ has the maximum number of \emph{relevant} vertices in $N(x)$, then $x$ is good to put in $P$.
Indeed, if $x$ dominates all relevant vertices adjacent to vertices of $S_i(s,t)$, then $x$ is a perfect choice to put in $P$.
Else, any vertex $y$ of a slice $S_i(s,t)$ is a good vertex to put in $P$. 
Hence, $P$ is optimal if the number of \emph{relevant} vertices adjacent to $P$ is maximal.
Thus, the path selected in line~\ref{line:dhgSelectPLoop} to~line~\ref{line:dhgSelectPAddV} is optimal.
\qed
\end{proof}

Running Algorithm~\ref{algo:MinEccDistHere} for a diametral pair of vertices of a distance-hereditary graph~$G$, by Theorem~\ref{th:dhg}, we get a shortest path of $G$ with minimum eccentricity.
Thus, we have proven the following result. 

\begin{theorem}
    \label{tm:opt-path} 
A shortest path with minimum eccentricity of a distance-hereditary graph $G=(V,E)$ can be computed in $\calO(|V| + |E|)$ total time.
\end{theorem}

\section{A Polynomial-Time Algorithm for Tree-Structured Graphs}

\subsection{Projection Gap}

In a graph~$G$, consider a shortest path~$P$ which starts in a vertex~$s$.
Each vertex~$x$ has a projection~$\Pi_x = \Pr(x, P)$.
In case of a tree this is a single vertex.
However, in general, $\Pi_x$ can contain multiple vertices and does not necessarily induce a connected subgraph.
In this case, there are two vertices $u$ and~$w$ in $\Pi_x$ such that all vertices~$v$ in the subpath~$Q$ between $u$ and~$w$ are not in~$\Pi_x$.
Formally, $u,w \in \Pi_x$, $Q = \{ \, v \in P \mid d_G(s,u) < d_G(s,v) < d_G(s, w) \}$, and $Q \cap \Pi_x = \emptyset$.

Now, assume the cardinality of $Q$ is at most~$\gamma$, \ie, $d_G(u,w) \leq \gamma + 1$ for each $P$, $x$, $u$ and~$w$.
We refer to $\gamma$ as the \emph{projection gap} of $G$.

\begin{definition}
    [%
        Projection Gap%
    ]
In a graph~$G$, let $P = \{ v_0, \ldots, v_l \}$ be a shortest path with $d_G(v_0, v_i) = i$.
The \emph{projection gap} of $G$ is $\gamma$, $\pg(G) = \gamma$ for short, if, for every vertex~$x$ of~$G$ and every two vertices $v_i, v_k \in \Pr(x, P)$, $d_G(v_i, v_k) > \gamma + 1$ implies that there is a vertex~$v_j \in \Pr(x, P)$ with $i < j < k$.
\end{definition}

Based on this definition, we can make the following observation.

\begin{lemma}
    \label{lem:valProp}
In a graph $G$ with $\pg(G) = \gamma$, let $P$ be a shortest path starting in~$s$, $Q$ be a subpath of~$P$, $|Q| > \gamma$, $u$ and~$w$ be two vertices in $P \setminus Q$ such that $d_G(s, u) < d_G(s,Q) < d_G(s, w)$, and $x$ be an arbitrary vertex in~$G$.
If $d_G(x, u) < d_G(x, Q)$, then $d_G(x, w) \geq d_G(x, Q)$.
\end{lemma}

\begin{proof}
Assume that $d_G(x, u) < d_G(x, Q)$ and $d_G(x, w) < d_G(x, Q)$.
Without loss of generality, let $d_G(x, u) = d_G(x, w) < d_G(x, v)$ for all $v \in P$ with $d_G(s,u) < d_G(s,v) < d_G(s, w)$.
Let $P'$ be the subpath of $P$ from $u$ to~$w$.
Note that $\Pr(x, P') = \{ u, w \}$ and $Q \subset P'$.
Thus, $d_G(u, w) \geq |Q| + 1 > \gamma + 1$.
This contradicts with $\pg(G) = \gamma$.
\qed
\end{proof}

Informally, Lemma~\ref{lem:valProp} says that, when exploring a shortest path~$P$, if the distance to a vertex~$x$ did not decrease during the last $\gamma + 1$ vertices of~$P$, it will not decrease when exploring the remaining subpath.
Based on this, we will show that a minimum eccentricity shortest path can be found in polynomial time if $\pg(G)$ is bounded by some constant.

For the rest of this section, we assume we are given a graph~$G$ with $\pg(G) = \gamma$ containing a vertex~$s$.
We will need the following notions and notations:
\begin{itemize}
\item
$Q_{i}$ and~$Q_j$ are subpaths of length~$\gamma$ of some shortest paths starting in~$s$.
They do not need to be subpath of the same shortest path.
Let $v_i \in Q_i$ and $v_j \in Q_j$ be the two vertices such that $d_G(s, Q_i) = d_G(s, v_i)$ and  $d_G(s, Q_j) = d_G(s, v_j)$.
Without loss of generality, let $d_G(s, v_i) \leq d_G(s, v_j)$.
We say, $Q_{i}$ is \emph{compatible} with $Q_j$ (with respect to~$s$) if $|Q_{i} \cap Q_j| = \gamma - 1$, $v_{i}$ is adjacent to~$v_j$, and $d_G(s, v_i) < d_G(s, v_j)$.
Let $\calC_s(Q_j)$ denote the set of subpaths compatible with~$Q_j$.

\item
$R_s(Q_j) = \{ \, w \mid Q_j \subseteq I(s, w) \} \cup Q_j$ is the set of vertices~$w$ such that there is a shortest path from $s$ to~$w$ containing $Q_j$ (or $w \in Q_j$).

\item
$I(s, Q_j) = I(s, v_j) \cup Q_j$ are the vertices that are on a shortest path from $s$ to~$Q_j$ (or in~$Q_j$).

\item
$V^\downarrow_s(Q_j) = \{ \, x \mid d_G(x, Q_j) = d_G(x, R_s(Q_j)) \}$ is the set of vertices~$x$ which are closer to~$Q_j$ than to all other vertices in~$R_s(Q_j)$.
Thus, given a shortest path~$P$ containing~$Q_j$ and starting in~$s$, expanding $P$ beyond $Q_j$ will not decrease the distance from $x$ to~$P$.
\end{itemize}
Note that $Q_j \subseteq V^\downarrow_s(Q_j)$ and $Q_j = I(s, Q_j) \cap R_s(Q_j)$.

\subsection{Algorithm}

\begin{lemma}
    \label{lem:kValeyVert}
For each vertex~$x$ in~$G$, $d_G(x, Q_j) = d_G(x, I(s, Q_j))$ or $d_G(x, Q_j) = d_G(x, R_s(Q_j))$.
\end{lemma}

\begin{proof}
Assume, $d_G(x, Q_j) > d_G(x, I(s, Q_j))$ and $d_G(x, Q_j) > d_G(x, R_s(Q_j))$.
Then, there is a vertex~$u_i \in I(s,Q_j)$ and a vertex~$u_r \in R_s(Q_j)$ with $d_G(x, u_i) < d_G(x, Q_j)$ and $d_G(x, Q_j) > d_G(x, u_r)$.
Because $u_i$, $Q_j$, and $u_r$ are on a shortest path starting in~$s$ and $|Q_j| > \gamma$, this contradicts Lemma~\ref{lem:valProp}.
\qed
\end{proof}

\begin{lemma}
    \label{lem:kValeyVertSubset}
If $Q_{i}$ is compatible with $Q_j$, then $V^\downarrow_s(Q_{i}) \subseteq V^\downarrow_s(Q_j)$.
\end{lemma}

\begin{proof}
Assume that $V^\downarrow_s(Q_{i}) \nsubseteq V^\downarrow_s(Q_j)$, \ie, there is a vertex~$x \in V^\downarrow_s(Q_{i}) \setminus V^\downarrow_s(Q_j)$.
Then, $d_G(x, Q_j) > d_G(x, R_s(Q_j))$.
Thus, by Lemma~\ref{lem:kValeyVert}, $d_G(x, Q_j) = d_G(x, I(s, Q_j))$.
Because $Q_{i} \subseteq I(s, Q_j)$, $d_G(x, Q_{i}) \geq d_G(x, I(s, Q_j)) = d_G(x, Q_j)$.
Since $x \in V^\downarrow_s(Q_{i})$, $d_G(x, Q_{i}) = d_G(x, R_s(Q_{i}))$.
Also, because $x \notin V^\downarrow_s(Q_{i})$, $d_G(x, Q_j) > d_G(x, R_s(Q_j))$.
Thus, $d_G(x, R_s(Q_{i})) > d_G(x, R_s(Q_j))$.
On the other hand, because $R_s(Q_{i}) \supseteq R_s(Q_j)$, $d_G(x, R_s(Q_{i})) \leq d_G(x, R_s(Q_j))$, and a contradiction arises.
\qed
\end{proof}

For a subpath $Q_j$, let $\calP_s(Q_j)$ denote the set of shortest paths~$P$ which start in~$s$ such that $Q_j \subseteq P \subseteq I(s, Q_j)$.
Then, we define $\varepsilon_s(Q_j)$ as follows:
\[
    \varepsilon_s(Q_j) = \min_{P \in \calP_s(Q_j)} \max_{x \in V^\downarrow_s(Q_j)} d_G(x, P)
\]
Consider a subpath $Q_j$ for which $R_s(Q_j) = Q_j$, \ie, a shortest path containing $Q_j$ cannot be extended any more.
Then, $V^\downarrow_s(Q_j) = V$.
Therefore, for any path $P \in \calP_s(Q_j)$, $\max_{x \in V^\downarrow_s(Q_j)} d_G(x, P) = \ecc(P)$.

\begin{lemma}
    \label{lem:valEpsilon}
If $\calC_s(Q_j)$ is not empty, then
\[
    \varepsilon_s(Q_j)
    =
    \min_{Q_{i} \in \calC_s(Q_j)} \max \left[
        \max_{x \in V^\downarrow_s(Q_j) \setminus V^\downarrow_s(Q_{i})} \min \big( d_G(x, Q_{i}), d_G(x, Q_{j}) \big),
        \varepsilon_s(Q_{i})
    \right].
\]
\end{lemma}

\begin{proof}
By definition,
\[
    \varepsilon_s(Q_j)
    =
    \min_{P \in \calP_s(Q_j)} \max_{x \in V^\downarrow_s(Q_j)} d_G(x, P).
\]
Let $Q_{i}$ be compatible with~$Q_j$.
Because, by Lemma~\ref{lem:kValeyVertSubset}, $V^\downarrow_s(Q_{i}) \subseteq V^\downarrow_s(Q_j)$, we can partition $V^\downarrow_s(Q_j)$ into $V^\downarrow_s(Q_j) \setminus V^\downarrow_s(Q_{i})$ and $V^\downarrow_s(Q_{i})$.
Thus, $\varepsilon_s(Q_j) =$
\[
    \min_{Q_{i} \in \calC_s(Q_j)} \min_{P \in \calP_s(Q_{i})} \max \left[
        \max_{x \in V^\downarrow_s(Q_j) \setminus V^\downarrow_s(Q_{i})} d_G(x, P \cup Q_j),
        \max_{x \in V^\downarrow_s(Q_{i})} d_G(x, P \cup Q_j)
    \right].
\]
Note that we changed the definition of~$P$ from $P \in \calP_s(Q_j)$ to $P \in \calP_s(Q_{i})$, \ie, $P$ may not contain the last vertex of $Q_j$ any more.

If $x \in V^\downarrow_s(Q_j) \setminus V^\downarrow_s(Q_{i})$, then $d_G(x, Q_{i}) > d_G(x, R_s(Q_{i}))$.
Thus, by Lemma~\ref{lem:kValeyVert}, $d_G(x, Q_{i}) = d_G(x, I(s, Q_{i}))$.
Note that, by definition of~$P$, $d_G(x, Q_i) \geq d_G(x, P) \geq d_G(x, I(s,Q_i))$.
Therefore, $d_G(x, P) = d_G(x, Q_i)$ and
\[
    \max_{x \in V^\downarrow_s(Q_j) \setminus V^\downarrow_s(Q_{i})} d_G(x, P \cup Q_j)
    =
    \max_{x \in V^\downarrow_s(Q_j) \setminus V^\downarrow_s(Q_{i})} \min \big( d_G(x, Q_{i}), d_G(x, Q_{j}) \big).
\]
For simplicity, we define
\[
    \varepsilon_s(Q_i, Q_j)
    :=
    \max_{x \in V^\downarrow_s(Q_j) \setminus V^\downarrow_s(Q_{i})} \min \big( d_G(x, Q_{i}), d_G(x, Q_{j}) \big).
\]
Note that $\varepsilon_s(Q_i, Q_j)$ does not depend on $P$.
Therefore, because $\min_{u} \max [ c, f(u) ] = \max [ c, \min_{u} f(u) ]$,
\[
    \varepsilon_s(Q_j)
    =
    \min_{Q_{i} \in \calC_s(Q_j)} \max \left[
        \varepsilon_s(Q_i, Q_j),
        \min_{P \in \calP_s(Q_{i})} \max_{x \in V^\downarrow_s(Q_{i})} d_G(x, P \cup Q_j)
    \right].
\]
If $x \in V^\downarrow_s(Q_{i})$, then $d_G(x, Q_{i}) = d_G(x, R_s(Q_{i})) \leq d_G(x, R_s(Q_j)) = d_G(x, Q_j)$.
Therefore,
\[
    \min_{P \in \calP_s(Q_{i})} \max_{x \in V^\downarrow_s(Q_{i})} d_G(x, P \cup Q_j)
    =
    \min_{P \in \calP_s(Q_{i})} \max_{x \in V^\downarrow_s(Q_{i})} d_G(x, P)
    =
    \varepsilon_s(Q_{i}).
\]
Thus,
\[
    \varepsilon_s(Q_j)
    =
    \min_{Q_{i} \in \calC_s(Q_j)} \max \left[
        \max_{x \in V^\downarrow_s(Q_j) \setminus V^\downarrow_s(Q_{i})} \min \big( d_G(x, Q_{i}), d_G(x, Q_{j}) \big),
        \varepsilon_s(Q_{i})
    \right].
\]
\qed
\end{proof}

Based on Lemma~\ref{lem:valEpsilon}, Algorithm~\ref{algo:valMESP} computes a shortest path starting in~$s$ with minimal eccentricity.
The algorithm has two parts.
First, it computes the pairwise distance of all vertices and $d_G(x, R_s(v))$ for each vertex pair $x$ and~$v$ where, similarly to $R_s(Q_j)$, $R_s(v) = \{ \, z \in V \mid v \in I(s, z) \}$.
This allows to easily determine if a vertex~$x$ is in~$V^\downarrow_s(Q_j)$.
Second, it computes $\varepsilon_s(Q_j)$ for each subpath~$Q_j$.
For this, the algorithm uses dynamic programming.
After calculating $\varepsilon_s(Q_i)$ for all subpaths with distance~$i$ to~$s$, the algorithm uses Lemma~\ref{lem:valEpsilon} to calculate $\varepsilon_s(Q_j)$ for all subpaths~$Q_j$ which $Q_i$ is compatible with.

\begin{algorithm}
    [htb!]
    \caption
    {
        Determines, for a given graph~$G$ with $\pg(G) \leq \gamma$ and a vertex~$s$, a minimal eccentricity shortest path starting in~$s$.
    }
    \label{algo:valMESP}

\KwIn
{
    A graph~$G = (V, E)$, an integer~$\gamma$, and a vertex~$s \in V$.
}

\KwOut
{
    A shortest path~$P$ starting in~$s$ with minimal eccentricity.
}

Determine the pairwise distances of all vertices.
\label{line:pairwDistance}

\ForEach
{
    $v, x \in V$
}
{
    Set $d_G(x, R_s(v)) := d_G(x, v)$.
    \label{line:defaultR}
}

\For
{
    $i = \ecc(s) - 1$ \KwDownTo $0$
    \label{line:compRsStart}
}
{
    \ForEach
    {
        $v \in L_i^{(s)}$
    }
    {
        \ForEach
        {
            $w \in N(v) \cap L_{i+1}^{(s)}$
        }
        {
            \ForEach
            {
                $x \in V$
            }
            {
                Set $d_G(x, R_s(v)) := \min \big[ d_G(x, R_s(v)), d_G(x, R_s(w)) \big]$.
                \label{line:distRv}
            }
        }
    }
}

\For
{
    $j = 0$ \KwTo $\ecc(s) - \gamma$
    \label{line:QjLoop}
}
{
    \ForEach
    {
        $Q_j$ with $d_G(s, Q_j) = j$
        \label{line:selectQj}
    }
    {
        \ForEach
        {
            $x \in V$
            \label{line:xQjLoop}
        }
        {
            Let $z_j$ be the vertex in $Q_j$ with the largest distance to~$s$.
            If $d_G(x, Q_j) \leq d_G(x, R_s(z_j))$, add $x$ to $V^\downarrow_s(Q_j)$ and store $d_G(x, Q_j)$.
            \label{line:computeVsOj}
        }

        \If
        {
            $j = 0$
        }
        {
            $\displaystyle \varepsilon_s(Q_j) := \max_{x \in V^\downarrow_s(Q_j)} d_G(x, Q_j)$
            \label{line:compStartEpsilon}
        }
        \Else
        {
            $\displaystyle \varepsilon_s(Q_j) := \infty$
            \label{line:QjEpsilonInf}
        }

        \ForEach
        {
            $Q_i \in \calC_s(Q_j)$
            \label{line:QiLoop}
        }
        {
            $\displaystyle
                \varepsilon'_s(Q_j) :=
                \max \left[
                    \max_{x \in V^\downarrow_s(Q_j) \setminus V^\downarrow_s(Q_i)} \min \big( d_G(x, Q_i), d_G(x, Q_j) \big),
                    \varepsilon_s(Q_i)
                \right]
            $
            \label{line:compEpsilonPrime}

            \If
            {
                $\varepsilon'_s(Q_j) < \varepsilon_s(Q_j)$
                \label{line:comareEpsilon}
            }
            {
                Set $\varepsilon_s(Q_j) := \varepsilon'_s(Q_j)$ and $p(Q_j) := Q_i$.
                \label{line:setP}
            }
        }
    }
}

Find a subpath~$Q_j$ such that a shortest path containing $Q_j$ cannot be extended any more and for which $\varepsilon_s(Q_j)$ is minimal.
\label{line:findMinQj}

Construct a path~$P$ from $Q_j$ to $s$ using the $p()$-pointers and output it.
\label{line:constructP}
\end{algorithm}

\begin{theorem}
For a given graph~$G$ with $\pg(G) = \gamma$ and a vertex~$s$, Algorithm~\ref{algo:valMESP} computes a shortest path starting in~$s$ with minimal eccentricity.
It runs in $\calO(n^{\gamma + 3})$ time if $\gamma \geq 2$, in $\calO(n^2m)$ time if $\gamma = 1$, and in $\calO(nm)$ time if $\gamma = 0$.
\end{theorem}

\begin{proof}
    [Correctness]
The algorithm has two parts.
The first part (line~\ref{line:pairwDistance} to line~\ref{line:distRv}) is a preprocessing which computes $d_G(x, R_s(v))$ for each vertex pair $x$ and~$v$.
The second part computes $\varepsilon_s(Q_j)$ which is then used to determine a path with minimal eccentricity.

For the first part, without loss of generality, let $d_G(s, v) = i$, $N^\uparrow_s(v) = N(v) \cap L_{i+1}^{(s)}$, and let $x$ be an arbitrary vertex.
By definition of $R_s$, $N^\uparrow_s(v) = \emptyset$ implies $R_s(v) = \{ v \}$, \ie, $d_G(x, R_s(v)) = d_G(x, v)$.
Therefore, $d_G(x, R_s(v))$ is correct for all vertices~$v$ with $N^\uparrow_s(v) = \emptyset$ after line~\ref{line:defaultR}.
By induction, assume that $d_G(x, R_s(w))$ is correct for all vertices~$w \in N^\uparrow_s(v)$.
Because $R_s(v) = \bigcup_{w \in N^\uparrow_s(v)} R_s(w) \cup \{ v \}$, $d_G(x, R_s(v)) = \min(\min_{w \in N^\uparrow_s(v)} d_G(x, R_s(w)), d_G(x, v))$.
Therefore, line~\ref{line:distRv} correctly computes~$d_G(x, R_s(v))$.

The second part of Algorithm~\ref{algo:valMESP} iterates over all subpaths~$Q_j$ in increasing distance to~$s$.
Line~\ref{line:computeVsOj} checks if a given vertex~$x$ is in $V^\downarrow_s(Q_j)$.
By definition, $R_s(Q_j) = Q_j \cup R_s(z_j)$ where $z_j$ is the vertex in $Q_j$ with the largest distance to~$s$.
Thus, $d_G(x, R_s(Q_j)) = \min( d_G(x, R_s(z_j)), d_G(x, Q_j))$.
By definition of $V^\downarrow_s$, $x \in V^\downarrow_s(Q_j)$ if and only if $d_G(x, Q_j) = d_G(x, R_s(Q_j))$.
Therefore, $x \in V^\downarrow_s(Q_j)$ if and only if $d_G(x, Q_j) \leq d_G(x, R_s(z_j))$, \ie, line~\ref{line:computeVsOj} computes $V^\downarrow_s(Q_j)$ correctly.

Recall the definition of $\varepsilon_s(Q_j)$:
\[
    \varepsilon_s(Q_j) = \min_{P \in \calP_s(Q_j)} \max_{x \in V^\downarrow_s(Q_j)} d_G(x, P)
\]
If $d_G(s, Q_j) = 0$, $d_G(x, P) = d_G(x, Q_j)$.
Therefore, $\varepsilon_s(Q_j) = \max_{x \in V^\downarrow_s(Q_j)} d_G(x, Q_j)$ as computed in line~\ref{line:compStartEpsilon}.
Note that there is no subpath~$Q_i$ which is compatible with $Q_j$, if $d_G(s, Q_j) = 0$.
Therefore, the loop starting in line~\ref{line:QiLoop} is skipped for these~$Q_j$.
Thus, the algorithm correctly computes $\varepsilon_s(Q_j)$, if $d_G(s, Q_j) = 0$.

By induction, assume that $\varepsilon_s(Q_i)$ is correct for each $Q_i \in \calC(Q_j)$.
Thus, Lemma~\ref{lem:valEpsilon} can be used to compute $\varepsilon_s(Q_j)$.
This is done in the loop starting in line~\ref{line:QiLoop}.
Therefore, at the beginning of line~\ref{line:findMinQj}, $\varepsilon_s(Q_j)$ is computed correctly for each subpath~$Q_j$.

Recall, if $P \in \calP(Q_j)$ and $R_s(Q_j) = Q_j$, then $V^\downarrow_s(Q_j) = V$ and, therefore, $\max_{x \in V^\downarrow_s(Q_j)} d_G(x, P) = \ecc(P)$.
Thus, $R_s(Q_j) = Q_j$ implies that $\varepsilon_s(Q_j)$ is the minimal eccentricity of all shortest paths starting in~$s$ and containing~$Q_j$.
Therefore, if $Q_j$ is picked by line~\ref{line:findMinQj}, then $\varepsilon_s(Q_j)$ is the minimal eccentricity of all shortest paths starting in~$s$.
\qed
\end{proof}

\begin{proof}
    [Complexity]
First, we will analyse line~\ref{line:pairwDistance} to line~\ref{line:distRv}.
Line~\ref{line:pairwDistance} runs in $\calO(nm)$ time.
This allows to access the distance between two vertices in constant time.
Thus, the total running time for line~\ref{line:defaultR} is~$\calO(n)$.
Because line~\ref{line:distRv} is called at most once for each vertex~$x$ and edge~$vw$, implementing line~\ref{line:compRsStart} to line~\ref{line:distRv} can be done in $\calO(nm)$ time.

For the second part of the algorithm (starting in line~\ref{line:QjLoop}), if $\gamma \geq 2$, let all subpaths be stored in a trie as follows:
There are $\gamma + 1$ layers of internal nodes.
Each internal node is an array of size~$n$ (one entry for each vertex) and each entry points to an internal node of the next layer representing $n$~subtrees.
This requires $\calO(n^{\gamma + 1})$ memory.
Leafs are objects representing a subpath.

If $\gamma = 1$, a subpath is a single edge, and, if $\gamma = 0$, a subpath is a single vertex.
Thus, no extra data structure is needed for these cases.
In all three cases, a subpath can be accessed in $\calO(\gamma)$ time.

Next, we analyse the runtime of line~\ref{line:xQjLoop} to line~\ref{line:QjEpsilonInf} for a single subpath~$Q_j$.
Accessing $Q_j$ can be done in $\calO(\gamma)$ time.
Line~\ref{line:computeVsOj} requires at most $\calO(\gamma)$ time for a single call and is called at most $\calO(n)$ times.
Line~\ref{line:compStartEpsilon} requires $\calO(n\gamma)$ time and line~\ref{line:QjEpsilonInf} runs in constant time.
Therefore, for a given subpath, line~\ref{line:xQjLoop} to line~\ref{line:QjEpsilonInf} require $\calO(\gamma n + n)$ time.

For line~\ref{line:compEpsilonPrime} to line~\ref{line:setP}, consider a given pair of compatible subpaths $Q_i$ and~$Q_j$.
Accessing both subpaths can be done in $\calO(\gamma)$ time.
Assuming the vertices in $V^\downarrow_s(Q_i)$ and~$V^\downarrow_s(Q_j)$ are sorted and stored with their distance to $Q_i$ and~$Q_j$, line~\ref{line:compEpsilonPrime} requires at most $\calO(n)$ time.
Note that $Q_i$ and $Q_j$ intersect in $\gamma - 1$ vertices.
Thus, $\min \big( d_G(x, Q_i), d_G(x, Q_j) \big) =  \min \big( d_G(x, v_i), d_G(x, Q_j) \big)$ where $v_i$ is the vertex in $Q_i$ closest to~$s$.
Line~\ref{line:comareEpsilon} and line~\ref{line:setP} run in constant time.
Therefore, for a given pair of compatible subpaths, line~\ref{line:compEpsilonPrime} to line~\ref{line:setP} require $\calO(n)$ time.

Let $\phi$ be the number of subpaths and $\psi$ be the number of pairs of compatible subpaths.
Then, the overall runtime for line~\ref{line:QjLoop} to line~\ref{line:setP} is $\calO(\phi (\gamma n + n) + \psi n)$ time, $\calO(\phi)$ time for line~\ref{line:findMinQj}, and $\calO(n)$ time for line~\ref{line:constructP}.
Together with the first part of the algorithm, the total runtime of Algorithm~\ref{algo:valMESP} is $\calO(mn + \phi (\gamma n + n) + \psi n)$.

Because a subpath contains $\gamma + 1$ vertices, there are up to $\calO(n^{\gamma + 1})$ subpaths and up to $\calO(n^{\gamma + 2})$ compatible pairs if $\gamma \geq 2$, \ie, $\phi \leq n^{\gamma + 1}$ and $\psi \leq n^{\gamma + 2}$.
Therefore, if $\gamma \geq 2$, Algorithm~\ref{algo:valMESP} runs in $\calO(n^{\gamma+3})$ time.

If $\gamma = 1$, a subpath is a single edge and there are at most $mn$ compatible pairs of subpaths, \ie, $\phi \leq m$ and $\psi \leq nm$.
For the case when $\gamma = 0$, a subpath is a single vertex ($\phi \leq n$) and a pair of compatible subpaths is an edge ($\psi \leq m$).
Therefore, Algorithm~\ref{algo:valMESP} runs in $\calO(n^2m)$ time if $\gamma = 1$, and in $\calO(nm)$ time if $\gamma = 0$.
\qed
\end{proof}

Note that Algorithm~\ref{algo:valMESP} computes a shortest path starting in a given vertex~$s$.
Thus, a shortest path with minimum eccentricity among all shortest paths in $G$ can be determined by running Algorithm~\ref{algo:valMESP} for all start vertices~$s$, resulting in the following:

\begin{theorem}
For a given graph~$G$ with $\pg(G) = \gamma$, a  minimum eccentricity shortest path can be found in $\calO(n^{\gamma + 4})$ time if $\gamma \geq 2$, in $\calO(n^3m)$ time if $\gamma = 1$, and $\calO(n^2m)$ time if $\gamma = 0$.
\end{theorem}

\subsection{Projection Gap for some Graph Classes}

Above, we have shown that a minimum eccentricity shortest path can be found in polynomial time if the projection gap is bounded by a constant.
In this subsection, we will determine the projection gap for some graph classes.

\subsubsection{Chordal Graphs and Dually Chordal Graphs.}

The class of chordal graphs is a well known class which can be recognised in linear time~\cite{TarjanYannak1984}.
Due to the strong tree structure of chordal graphs, they have the following property known as $m$-convexity:

\begin{lemma}
     [\cite{FaberJamison1986}]
     \label{lem:ChordalMConvex}
Let $G$ be a chordal graph.
If, for two distinct vertices~$u,v$ in a disk~$N^r[x]$, there is a path~$P$ connecting them with $P \cap N^r[x] = \{ u, v \}$, then $u$ and~$v$ are adjacent.
\end{lemma}

\begin{lemma}
If $G$ is a chordal graph, then $\pg(G) = 0$.
\end{lemma}

\begin{proof}
Assume $\pg(G) \geq 1$.
Then, there is a shortest path~$P = \{ u, \ldots, w \}$ and a vertex~$x$ with $\Pr(x, P) = \{ u, w \}$ and $d_G(u, w) > 1$.
By Lemma~\ref{lem:ChordalMConvex}, $u$ and $w$ are adjacent.
This contradicts with $d_G(u, w) > 1$.
\qed
\end{proof}

\begin{corollary}
For chordal graphs, a minimum eccentricity shortest path can be found in $\calO(n^2m)$ time.
\end{corollary}

Dually chordal graphs where introduced in~\cite{BraDraCheVol1998}.
They are closely related to chordal graphs.

\begin{lemma}
If $G$ is a dually chordal graph, then $\pg(G) \leq 1$.
\end{lemma}

\begin{proof}
Assume there is a shortest path~$P = \{ u, v_1, \ldots, v_i , w \}$ and a vertex~$x$ with $\Pr(x, P) \supseteq \{ u, w \}$.
To show that $\pg(G) \leq 1$, we will show that $d_G(u, w) = i + 1 > 2$ implies there is a vertex~$v_k \in \Pr(x, P)$ with $1 \leq k \leq i$.

Consider a family of disks $\calD = \big \{ N[u], N[v_{1}], \ldots, N[v_{i}], N[w], N^r[x]  \big \}$ where $r = d_G(x, P) - 1$.
Let $H$ be the intersection graph of $\calD$, $a$ be the vertex in~$H$ representing $N[u]$, $b_k$ representing $N[v_{k}]$ (for $1 \leq k \leq i$), $c$ representing $N[w]$, and $z$ representing~$N^r[x]$.
Because the intersection graph of disks of a dually chordal graph is chordal~\cite{BraDraCheVol1998}, $H$ is chordal, too.
$H$ contains the edges $za$ and~$zc$, $ab_1$, $cb_i$, and $b_kb_{k+1}$ for all $1 \leq k < i$.
Note that, if $d_G(u, w) > 2$, $a$ and~$c$ are not adjacent in~$H$.
However, the path $\{ a, b_1, \ldots, b_k, c \}$ connects $a$ and~$c$.
Therefore, because $H$ is chordal and by Lemma~\ref{lem:ChordalMConvex}, there is a $k$ with $1 \leq k \leq i$ such that $z$ is adjacent to~$b_k$ in~$H$.
Thus, $d_G(x, v_k) \leq r + 1$, \ie, $v_k \in \Pr(x, P)$.
\qed
\end{proof}

\begin{corollary}
For dually chordal graphs, a minimum eccentricity shortest path can be found in $\calO(n^3m)$ time.
\end{corollary}

\subsubsection{Graphs with bounded Tree-Length or Tree-Breadth.}

As defined by \textsc{Robertson} and \textsc{Seymour}~\cite{RobertSeymou1986}, a \emph{tree-decomposition} of a graph $G = (V, E)$ is a tree $\calT$ with the vertex set~$\calB$ where each vertex of $\calT$, called bag, is a subset of~$V$ such that:
(i)~$V = \bigcup_{B \in \calB} B$, (ii)~for each edge~$uv \in E$, there is a bag~$B \in \calB$ with $u,v \in B$, and (iii)~for each vertex~$v \in V$, the bags containing $v$ induce a subtree of~$\calT$.

The \emph{length} of a tree~decomposition is smaller than or equal to~$\lambda$ if for each bag~$B$, $\diam_G(B) \leq \lambda$.
A graph~$G$ has \emph{tree-length}~$\lambda$, if there exist a tree-decomposition~$\calT$ for~$G$ such that $\calT$ has length~$\lambda$.
Similarly, the \emph{breadth} of a tree~decomposition is smaller than or equal to~$\rho$ if for each bag~$B$ there is a vertex~$v \in V$ with $N^\rho[v] \supseteq B$.
A graph~$G$ has \emph{tree-breadth}~$\rho$, if there exist a tree-decomposition~$\calT$ for~$G$ such that $\calT$ has breadth~$\lambda$.

For these graphs, we use a concept called \emph{layering partition}.
It was introduced in~\cite{BranChepDrag1999,ChepoiDragan2000}.
The idea is to first partition the vertices of a given graph in distance layers~$L_i^{(x)}$ with respect to a given vertex~$x$.
Second, partition each layer~$L_i^{(x)}$ into \emph{clusters} in such a way that two vertices $u$ and~$v$ are in the same cluster if they are connected by a path~$P$ such that $d_G(x, P) = d_G(x, u)$, \ie, $P$ does not contain vertices of layers closer to~$x$ than $u$ and~$v$.

Unfortunately, computing the tree-length of a graph is an NP-hard problem~\cite{Lokshtanov2010}.
However, for our needs, an approximation of it would suffice.

\begin{lemma}
If $G$ has tree-length~$\lambda$ or tree-breadth~$\rho$, a factor~$\gamma \geq \pg(G)$ can be computed in $\calO(n^3)$ time such that $\gamma \leq 3 \lambda - 1$ or $\gamma \leq 6 \rho - 1$, respectively.
\end{lemma}

\begin{proof}
To compute $\gamma$, first determine the pairwise distances of all vertices.
Then, compute a layering partition for each vertex~$x$.
Let $\gamma + 1$ be the maximum diameter of all clusters of all layering partitions.

The diameter of each cluster is at most $3 \lambda$ if $G$ has tree-length~$\lambda$ and at most $6 \rho$ if $G$ has tree-breadth~$\rho$~\cite{DouDraGavYan2007,DraganKohler2014}.
Therefore, for each shortest path~$P$, $\diam(\Pr(x, P)) \leq 3 \lambda$ and $\diam(\Pr(x, P)) \leq 6 \rho$, respectively.
Thus, $\pg(G) \leq \gamma \leq 3 \lambda - 1$ and $\pg(G) \leq \gamma \leq 6 \rho - 1$.

Computing the pairwise distances of all vertices can be done in $\calO(nm)$ time.
A layering partition can be computed in linear time~\cite{ChepoiDragan2000}.
For a given layering partition, the diameter of each cluster can be computed in $\calO(n^2)$ time if the pairwise distances of all vertices are known.
Thus, $\gamma$ can be computed in $\calO(n^3)$ time.
\qed
\end{proof}

Note that it is not necessary to know the tree-length or tree-breath of~$G$ to compute~$\gamma$.
Thus, by computing~$\gamma$ and then running Algorithm~\ref{algo:valMESP} for each vertex in $G$, we get:

\begin{corollary}
For graphs with tree-length~$\lambda$ or tree-breadth~$\rho$, a minimum eccentricity shortest path can be found in $\calO(n^{3\lambda+3})$ time or $\calO(n^{6\rho+3})$ time, respectively.
\end{corollary}

\subsubsection{$\delta$-Hyperbolic Graphs.}

A graph has \emph{hyperbolicity}~$\delta$ if for any four vertices $u$, $v$, $w$, and~$x$, the  two  larger of the sums $d_G(u, v) + d_G(w, x)$, $d_G(u, w) + d_G(v, x)$, and $d_G(u, x) + d_G(v, w)$ differ by at most~$2\delta$.

\begin{lemma}
    [\cite{CheDraEstHab2008}]
    \label{lem:deltaHyper}
Let $u$, $v$, $w$, and $x$ be four vertices in a $\delta$-hyperbolic graph.
If $d_G(u, w) > \max \{ d_G(u, v), d_G(v, w) \} + 2\delta$, then $d_G(v, x) < \max \{ d_G(x, u), d_G(x, w) \}$.
\end{lemma}

\begin{lemma}
If $G$ is $\delta$-hyperbolic, then $\pg(G) \leq 4 \delta$.
\end{lemma}

\begin{proof}
Consider two vertices $u$ and~$w$ such that $u,w \in \Pr(x, P)$ for some vertex~$x$ and shortest path~$P$.
Let $v \in P$ be a vertex such that $d_G(u, v) - d_G(v, w) \leq 1$ and $d_G(u, v) \geq d_G(v, w)$, \ie, $v$ is a middle vertex on the subpath from $u$ to~$w$.

Assume, $d_G(u, w) > 4\delta + 1$.
Thus, $d_G(u, v) \geq d_G(v, w) \geq 2\delta + 1$ and $d_G(u, w) > d_G(u, v) + 2\delta$.
Therefore, by Lemma~\ref{lem:deltaHyper}, $d_G(v, x) < \max \{ d_G(x, u), d_G(x, w) \}$.
This contradicts that $u,w \in \Pr(x, P)$.
Hence, the diameter of a projection is at most $4 \delta + 1$ and, therefore, $\pg(G) \leq 4 \delta$.
\qed
\end{proof}

\begin{corollary}
For $\delta$-hyperbolic graphs, a minimum eccentricity shortest path can be found in $\calO(n^{4\delta+4})$ time.
\end{corollary}

\section{Approximation for Graphs with Bounded Tree-Length and Bounded Hyperbolicity}

In the last sections, we have shown how to find a shortest path with minimum eccentricity~$k$ for several graph classes.
For graphs with tree-length~$\lambda$, this can require up to $\calO(n^{3\lambda + 3})$ time.
In this section, we will show that, for graphs with tree-length~$\lambda$, we can find a shortest path with eccentricity at most $k + 2.5\lambda$ in at most $\calO(\lambda m)$ time and, for graphs with hyperbolicity~$\delta$, we can find a shortest path with eccentricity at most $k + \calO(\delta \log n)$ in at most $\calO(\delta m)$~time.

\begin{lemma}
    \label{lem:mutualDistHyper}
Let $G$ be a graph with hyperbolicity~$\delta$.
Two vertices $x$ and~$y$ in $G$ with $\ecc(x) = \ecc(y) = d_G(x, y)$ can be found in $\calO(\delta m)$ time.
\end{lemma}

\begin{proof}
Let $u$ and~$v$ be two vertices in $G$ such that $d_G(u, v) = \diam(G)$.
For an arbitrary vertex~$x_0$ and for $i \geq 0$, let $y_i = x_{i+1}$ be vertices in~$G$ such that $d_G(x_i, y_i) = \ecc(x_i)$ and $d_G(x_i, y_i) < d_G(x_{i+1}, y_{i+1})$.
To prove Lemma~\ref{lem:mutualDistHyper}, we will show that there is no vertex~$y_{2\delta + 1}$.

Because $d_G(x_0, y_0) = \ecc(x_0)$, $d_G(x_0, y_0) \geq \max \{ d_G(x_0, u), d_G(x_0 ,v) \}$.
Therefore, by Lemma~\ref{lem:deltaHyper}, $d_G(u, v) \leq \max \{ d_G(u, y_0), \allowbreak d_G(v, y_0) \} + 2\delta$ and, thus, $\diam(G) \leq \ecc(x_1) + 2\delta$.
Since $d_G(x_i, y_i) < d_G(x_{i+1}, y_{i+1})$, there is no vertex~$y_j$ with $j \geq 2\delta + 1$, otherwise $d_G(x_j, y_j) > \diam(G)$.
Therefore, a vertex pair~$x,y$ with $\ecc(x) = \ecc(y) = d_G(x, y)$ can be found in $\calO(\delta m)$~time as follows:
Pick an arbitrary vertex~$x_0$ and find a vertex~$x_1$ with $d_G(x_0, x_1) = \ecc(x_0)$ using a BFS.
Next, find a vertex~$x_2$ such that $d_G(x_1, x_2) = \ecc(x_1)$.
Repeat this (at most $2\delta$ times) until $d_G(x_i, x_{i+1}) = \ecc(x_i) = \ecc(x_{i+1})$.
\qed
\end{proof}

Note that, if a graph has tree-length~$\lambda$, its hyperbolicity is at most $\lambda$~\cite{CheDraEstHab2008}.
Thus, it follows:

\begin{corollary}
    \label{cor:mutualDistTreeLen}
Let $G$ be a graph with tree-length~$\lambda$.
Two vertices $x$ and~$y$ in $G$ with $\ecc(x) = \ecc(y) = d_G(x, y)$ can be found in $\calO(\lambda m)$ time.
\end{corollary}

The next lemma will show that, in a graph with bounded tree-length, a shortest path between two mutually furthest vertices gives an approximation for the MESP-problem.

\begin{lemma}
    \label{lem:ApproxTreeLen}
Let $G$ be a graph with tree-length~$\lambda$ having a shortest path with eccentricity~$k$.
Also, let $x$ and~$y$ be two mutually furthest vertices, \ie, $\ecc(x) = \ecc(y) = d_G(x, y)$.
Each shortest path from $x$ to~$y$ has eccentricity less than or equal to $k + 2.5 \lambda$.
\end{lemma}

\begin{proof}
Let $P$ be a shortest path from $s$ to~$t$ with eccentricity~$k$ and $Q$ be a shortest path from $x$ to~$y$.
Consider a tree-decomposition~$\calT$ for $G$ with length~$\lambda$.
We distinguish between two cases:
(1)~There is a bag in~$\calT$ containing a vertex of $P$ and a vertex of $Q$ and (2)~there is no such bag in~$\calT$.

\paragraph{Case 1: There is a bag in~$\calT$ containing a vertex of $P$ and a vertex of~$Q$.}
We define bags $B_x$ and $B_y$ as follows:
Both contain a vertex of~$P$ and a vertex of~$Q$, $B_x$ is a bag closest to a bag containing~$x$, $B_y$ is a bag closest to a bag containing~$y$, and the distance between $B_x$ and~$B_y$ in~$\calT$ is maximal.
Let $\{ B_0, B_1, \ldots, B_l \}$ be a subpath of the shortest path from $B_x$ to~$B_y$ in~$\calT$ such that $B_0$ is a bag closest to a bag containing~$s$, $B_l$ is a bag closest to a bag containing~$t$, $B_i$ is adjacent to~$B_{i+1}$ in $\calT$ ($0 \leq i < l$), and the distance~$l$ between $B_0$ and~$B_l$ is maximal.
Without loss of generality, let $d_{\calT}(B_x, B_0) \leq d_{\calT}(B_x, B_l)$.
Let $p_s$ be the vertex in $B_0 \cap P$ which is closest to~$s$ in~$G$ and let $p_t$ be the vertex in $B_l \cap P$ which is closest to~$t$ in~$G$.
Figure~\ref{fig:TreeDecoEg} gives an illustration.

\begin{figure}
    [htb]
    \centering
    \includegraphics[]{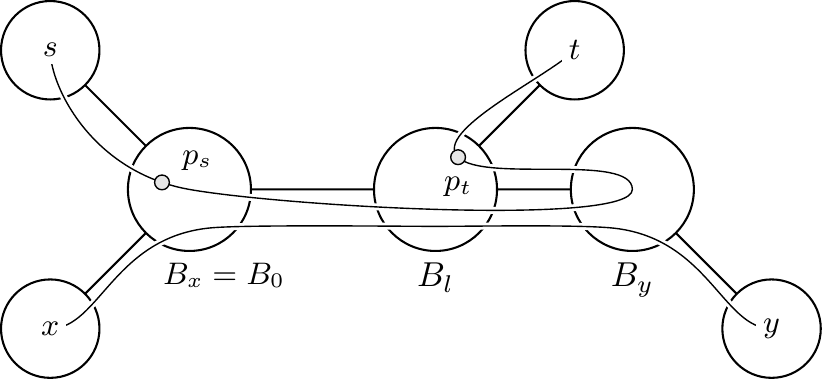}%
    \caption
    {
        Example for a possible tree-decomposition.
    }
    \label{fig:TreeDecoEg}
\end{figure}

\begin{claim}
For each vertex~$p \in P$ with $d_G(s, p_s) \leq d_G(s, p) \leq d_G(s, p_t)$, $d_G(p, Q) \leq 1.5 \lambda$.
\end{claim}

\begin{proof}
    [Claim]
There is a vertex set $\{ p_s = p_0, p_1, \ldots, p_l, p_{l+1} = p_t \} \subseteq P$, where $p_i \in B_{i-1} \cap B_i$ for all positive $i \leq l$.
Because $p_i, p_{i+1} \in B_i$ for $0 \leq i \leq l$, $d_G(p_i, p_{i+1}) \leq \lambda$.
Thus, because $P$ is a shortest path, for all $p' \in P$ with $d_G(s, p_s) \leq d_G(s, p') \leq d_G(s, p_t)$ there is a vertex~$p_i$ with $0 \leq i \leq l+1$ such that $d_G(p_i, p') \leq 0.5 \lambda$.
By definition of $\calT$, each bag~$B_i$ ($0 \leq i \leq l$) contains a vertex~$q \in Q$, \ie, $d_G(p_i, Q) \leq \lambda$ ($0 \leq i \leq l+1$).
Therefore, for all $p' \in P$ with $d_G(s, p_s) \leq d_G(s, p') \leq d_G(s, p_t)$ there is a vertex~$p_i$ with $0 \leq i \leq l+1$ such that $d_G(p', Q) \leq d_G(p_i, p') + d_G(p_i, Q) \leq 1.5 \lambda$.
\qedClaim
\end{proof}

Consider an arbitrary vertex~$v$ in~$G$.
Let $v'$ be a vertex in~$P$ closest to~$v$ and let $P_v$ be a shortest path from $v$ to~$v'$.
If $v'$ is between $p_s$ and $p_t$, \ie, $d_G(s, p_s) \leq d_G(s, v') \leq d_G(s, p_t)$, by the claim above, $d_G(v, Q) \leq d_G(v, v') + d_G(v', Q) \leq k + 1.5 \lambda$.
If $P_v$ intersects a bag containing a vertex~$q \in Q$, $d_G(v, Q) \leq k + \lambda$.

Next, consider the case when $P_v$ does not intersect a bag containing a vertex of $Q$ and (without loss of generality) $d_G(s, v') > d_G(s, p_t)$.
In this case, each path from $x$ to~$v$ intersects~$B_l$.

\begin{claim}
There is a vertex~$u \in B_l$ such that $d_G(u, y) \leq k + 0.5 \lambda$.
\end{claim}

\begin{proof}
    [Claim]
Let $y'$ be a vertex in~$P$ that is closest to~$y$ and let $P_y$ be a shortest path from $y$ to~$y'$.
If $P_y$ intersects~$B_l$, there is a vertex $u \in P_y \cap B_l$ with $d_G(y, u) \leq k$.

If $P_y$ does not intersect~$B_l$, there is a subpath of~$P$ starting at~$p_t$, containing~$y'$, and ending in a vertex~$p_l \in B_l$.
Because $d_G(p_t, p_l) \leq \lambda$, $d_G(y', \{ p_t, p_l \}) \leq 0.5 \lambda$.
Therefore, $d_G(y, \{ p_t, p_l \}) \leq d_G(y, y') + d_G(y', \{ p_t, p_l \}) \leq k + 0.5 \lambda$.
\qedClaim
\end{proof}

Let $u$, $v'$, and~$z$ be vertices in~$B_l$ such that $d_G(u, y) \leq k + 0.5 \lambda$, $v'$ is on a shortest path from $x$ to~$v$, and $z \in Q$.
Because $d_G(x, y) = \ecc(x)$, $d_G(x, v') + d_G(v', v) \leq d_G(x, y)$.
Also, by the triangle inequality, $d_G(x, y) \leq d_G(x, v') + d_G(v', y)$ and $d_G(v', y) \leq d_G(v', u) + d_G(u, y)$.
Because $\{ u, v', z \} \subseteq B_l$ and $d_G(u, y) \leq k + 0.5 \lambda$, $d_G(v', v) \leq k + 1.5 \lambda$ and therefore $d_G(z, v) \leq k + 2.5 \lambda$.

Thus, if there is a bag in~$\calT$ containing a vertex of~$P$ and a vertex of~$Q$, $d_G(v, Q) \leq k + 2.5 \lambda$ for all vertices~$v$ in~$G$.

\paragraph{Case 2: There is no bag in~$\calT$ containing vertices of $P$ and~$Q$.}
Because there is no such bag, $\calT$ contains a bag~$B$ such that each path from $x$ and~$y$ to~$P$ intersects~$B$ and there is a vertex~$z \in B \cap Q$.

Consider an arbitrary vertex~$v$.
If there is a shortest path~$P_v$ from $v$ to~$P$ which intersects~$B$, then $d_G(z, v) \leq k + \lambda$.
If there is no such path, each path from $x$ to~$v$ intersects~$B$.
Let $v' \in B$ be a vertex on a shortest path from $x$ to~$v$ and let $u \in B$ be a vertex on a shortest path from $y$ to~$P$.
Note that $d_G(u, y) \leq k$.

Because $d_G(x, y) = \ecc(x)$, $d_G(x, v') + d_G(v', v) \leq d_G(x, y)$.
Also, by the triangle inequality, $d_G(x, y) \leq d_G(x, v') + d_G(v', y)$ and $d_G(v', y) \leq d_G(v', u) + d_G(u, y)$.
Because $\{ u, v', z \} \subseteq B$ and $d_G(u, y) < k$, $d_G(v', v) < k + \lambda$ and therefore $d_G(z, v) < k + 2 \lambda$.

Thus, if there is no bag in~$\calT$ containing vertices of~$P$ and~$Q$, $d_G(v, Q) < k + 2 \lambda$ for all vertices~$v$ in~$G$.
\qed
\end{proof}

In~\cite{CheDraEstHab2008}, it was shown that an $n$-vertex $\delta$-hyperbolic graph has tree-length at most~$\calO(\delta \log n)$.

\begin{corollary}
    \label{cor:ApproxHyperbolic}
Let $G$ be a graph with hyperbolicity~$\delta$ having a shortest path with eccentricity~$k$.
Also, let $x$ and~$y$ be two mutually furthest vertices, \ie, $\ecc(x) = \ecc(y) = d_G(x, y)$.
Each shortest path from $x$ to~$y$ has eccentricity less than or equal to k + $\calO(\delta \log n)$.
\end{corollary}

Lemma~\ref{lem:mutualDistHyper}, Lemma~\ref{lem:ApproxTreeLen}, Corollary~\ref{cor:mutualDistTreeLen}, and Corollary~\ref{cor:ApproxHyperbolic} imply our main result of this section:

\begin{theorem}
Let $G$ be a graph having a shortest path with eccentricity~$k$.
If $G$ has tree-length~$\lambda$, a shortest path with eccentricity at most $k + 2.5\lambda$ can be found in $\calO(\lambda m)$~time.
If $G$ has hyperbolicity~$\delta$, a shortest path with eccentricity at most $k + \calO(\delta \log n)$ can be found in $\calO(\delta m)$~time.
\end{theorem}

A graph is chordal if and only if it has tree-length~$1$~\cite{Gavril1974}.

\begin{corollary}
If $G$ is a chordal graph and has a shortest path with eccentricity~$k$, a shortest path in~$G$ with eccentricity at most $k + 2$ can be found in linear time.
\end{corollary}

Figure~\ref{fig:chordalEg} gives an example that, for chordal graphs, $k + 2$ is a tight upper bound for the eccentricity of the determined shortest path.

\begin{figure}
    [htb]
    \centering
    \includegraphics[]{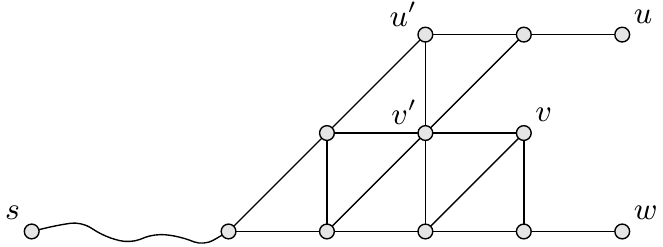}%
    \caption
    {
        A chordal graph~$G$.
        A shortest path from $s$ to~$v$ passing $v'$ has eccentricity~$2$ which is the minimum for all shortest paths in~$G$.
        The diametral path from $s$ to~$u$ passing $u'$ has eccentricity~$4$ because of its distance to~$w$.
    }
    \label{fig:chordalEg}
\end{figure}

\section{Conclusion}

We have investigated the Minimum Eccentricity Shortest Path problem for some structured graph classes.
For these classes, we were able to present linear or polynomial time algorithms.
Additionally, we presented a simple algorithm which gives an additve approximation in linear time for chordal graphs, in $\calO(\lambda m)$ time for graphs with tree-length~$\lambda$, and in $\calO(\delta m)$ time for graphs with hyperbolicity~$\delta$.
Table~\ref{tbl:Results} gives an overview of our results.

\begin{table}
    [htb]
    \caption
    {%
        Runtime for solving the Minimum Eccentricity Shortest Path problem for some graph classes. Also, if the solution is not optimal, the maximal difference to an optimal solution is shown.
    }
    \label{tbl:Results}
    \centering
    \begin{tabular}{@{\quad}l@{\quad}||@{\quad}l@{\quad}@{\quad}||@{\quad}l@{\quad}}
        Graph class & Runtime & Approx. \\
        \hline
        \hline
        distance-hereditary & linear \\
        \hline
        chordal & $\calO(n^2m)$ \\
                & linear & $+ 2$ \\
        \hline
        dually chordal & $\calO(n^3m)$ \\
        \hline
        tree-length~$\lambda$ & $\calO(n^{3\lambda+3})$ \\
                              & $\calO(\lambda m)$  & $+ 2.5 \lambda$ \\
        \hline
        tree-breadth~$\rho$ & $\calO(n^{6\rho+3})$ \\
        \hline
        $\delta$-hyperbolic & $\calO(n^{4\delta+4})$ \\
                            & $\calO(\delta m)$ & $+ \calO(\delta \log n)$\\
    \end{tabular}

\end{table}

One reason why the runtime to find an optimal path for distance-hereditary graphs is linear is that we can determine the start and end vertices of an optimal path in linear time for these graphs.
For the other classes, the algorithm iterates over all possible start vertices~$s$.
We know that, for general graphs, the problem remains NP-complete even if a start-end vertex pair is given (see the reduction in~\cite{DrLei2015}).
Also, we have shown that there is a shortest path with minimum eccentricity between every diametral pair of vertices of a distance-hereditary graph (Theorem~\ref{th:dhg}).
This leads to the following question:
How hard is it to determine the start and end vertices of an optimal path?
This question applies to general graphs as well as to special graph classes like chordal graphs.

Another interesting question is, for which other graph classes the problem remains NP-complete or can be solved in polynomial time.
The NP-completeness proof in~\cite{DrLei2015} uses a reduction from SAT.
There is a planar version of 3-SAT (see~\cite{Lichtenstein1982}).
Does this imply that the problem remains NP-complete for planar graphs?

\subsubsection{Acknowledgement:}
This work was partially supported by the NIH grant R01 GM103309.

\end{document}